\documentclass{article}
\usepackage{amsmath,amsfonts,amssymb,amsthm,graphicx}
 \usepackage{dsfont}    
\usepackage{enumerate}
\usepackage{color}
\usepackage{natbib}
\usepackage{multirow}
\usepackage{comment}

\newcommand{\ntp}{\mathbf{P}}
\newcommand{\ent}{\mathrm{Ent}}

\newcommand{\bP}{\mathbf{P}}

\newcommand{\uu}{\mathbb{U}}

\usepackage{booktabs,subcaption,dcolumn}

\newcommand{\lawis}{\buildrel \mathrm{d} \over \sim}

\usepackage{algorithm}
\usepackage{algorithmic}
\usepackage[colorlinks=true,urlcolor=black,citecolor=blue,pdfpagemode=UseNone,pdfstartview=FitH]{hyperref}

\newif\ifFULL
\ifFULL

\fi

\newif\ifRuodu
\ifRuodu
  \renewcommand{\ge}{\geqslant}
  \renewcommand{\le}{\leqslant}
  \renewcommand{\epsilon}{\varepsilon}
  
\fi

\renewcommand{\d}{\mathrm{d}}

\newcommand{\p}{\mathbb{P}}
\newcommand{\q}{\mathbb{Q}}

\newcommand{\E}{\mathbb{E}}    
\newcommand{\R}{\mathbb{R}}    
\newcommand{\N}{\mathbb{N}}    
\newcommand{\T}{\mathbb{T}_0} 
\newcommand{\TP}{\mathbb{T}}

  \newcommand{\id}{\mathds{1}}

\usepackage{setspace}

\theoremstyle{plain}
\newtheorem{theorem}{Theorem} 

\newtheorem{lemma}{Lemma}
\newtheorem{proposition}{Proposition}

\theoremstyle{definition}

\newtheorem{assumption}{Assumption}

\theoremstyle{remark}
\newtheorem{remark}{Remark}

 \renewcommand{\cite}{\citet}  
\DeclareMathOperator*{\argmax}{arg\,max}

\title{Online LLM watermark detection via e-processes} 
\author{Weijie Su\thanks{Department of Statistics and Data Science, University of Pennsylvania, Philadelphia, USA.    E-mail: \href{mailto:suw@wharton.upenn.edu}{suw@wharton.upenn.edu}}\and 
    Ruodu Wang\thanks {Department of Statistics and Actuarial Science,
   University of Waterloo,
   Waterloo, Ontario, Canada.
   E-mail: \href{mailto:wang@uwaterloo.ca}{wang@uwaterloo.ca}}
  \and Zinan Zhao\thanks{Center for Data Science and School of Mathematical Sciences, Zhejiang University, Hangzhou, China. E-mail: \href{mailto:znzh@zju.edu.cn}{znzh@zju.edu.cn}}
  }
 \date{\today}

\begin{document} 
\maketitle 
\makeatletter
{
  \renewcommand{\thefootnote}{}          
  \addtocounter{footnote}{1}     
  \footnotetext{Authors are listed in alphabetical order.}
  \addtocounter{footnote}{-1}            
}
\makeatother

\begin{abstract}
Watermarking for large language models (LLMs) has emerged as an effective tool for distinguishing AI-generated text from human-written content.
Statistically, watermark schemes induce dependence between generated tokens and a pseudo-random sequence, reducing watermark detection to a hypothesis testing problem on independence.
We develop a unified framework for LLM watermark detection based on e-processes, providing  anytime-valid guarantees for online testing.  
We propose various methods to construct empirically adaptive e-processes that can enhance the detection power. 
The proposed methods are applicable to any sequential testing problem where independent pivotal statistics are available.   In addition, theoretical results are established to characterize the power properties of the proposed procedures. Some experiments demonstrate that the proposed framework achieves competitive performance compared to existing watermark detection methods.\medskip  
\end{abstract}

\section{Introduction}

The recent advent of powerful Large Language Models (LLMs) has fundamentally transformed natural language generation. State-of-the-art LLMs such as GPT-4 \citep{openai2024gpt4} and LLaMA \citep{meta2024llama3} demonstrate remarkable capabilities in producing human-like text. While these advances unlock substantial societal and economic benefits, they also raise significant ethical and security concerns. The increasing difficulty in distinguishing AI-generated text from human-written content facilitates its misuse for malicious purposes, including generation of misinformation \citep{zellers2019defending,starbird2019disinformation},
academic plagiarism \citep{stokel2022ai,milano2023edu}, and automated fraudulent content generation \citep{shumailov2023curse,das2024llm}. Such misuse may seriously erode trust in digital ecosystems. To mitigate these problems, there is a pressing need for effective mechanisms to make AI-generated text distinguishable while preserving the quality and usability of language models.

Watermarking has emerged as a promising approach to address this challenge \citep{fernandez2023wm,kuditipudi2024robust,hu2024unbiased}. In watermark schemes for LLMs, a subtle algorithmic signal is embedded into the text generation process, enabling post-hoc verification of AI authorship without significantly degrading text quality. From a statistical perspective, watermark schemes modify the sampling procedure of the language model.

Concretely, let $\mathcal{W}$ denote the vocabulary of tokens. Given a sequence of previous tokens $(w_1,\ldots,w_{t-1})$, an unwatermarked language model computes a next-token prediction (NTP) distribution $\ntp_t$ and samples $w_t$ from $\ntp_t$. In contrast, a watermarked model generates $w_t$ according to both $\ntp_t$ and a pseudo-random variable (or vector) $\zeta_t$. {In practice, the pseudo-random sequence can be reconstructed using the generated text together with shared secret information, referred to as the watermark key.} A formal description of this mechanism is deferred to Section~\ref{sec:formulation}.
Crucially, the watermark mechanism induces statistical dependence between $w_t$ and $\zeta_t$. In unwatermarked text generation, $w_t$ is independent of $\zeta_t$, whereas in watermarked generation, this independence is deliberately violated. Therefore, watermark detection can be formulated as a hypothesis testing problem on the dependence structure between the observed tokens and the pseudo-random variables \citep{aar2023,aos2025wm,li2025robust,markmyword2025}.



Despite substantial progress, existing watermark detection methods still face several challenges. First, most procedures are designed for fixed-sample testing and lack rigorous guarantees under optional stopping. This limitation is particularly problematic in real-world implementations, where LLMs generate text as a stream. In such sequential settings, repeatedly checking for evidence of the watermark using traditional p-value methods may inflate false positive rates, yet waiting for the full text to be generated introduces high computational costs and unacceptable latency for applications requiring immediate intervention. This challenge is especially critical in the emerging context of autonomous agents; as AI tools like OpenClaw automate workflows in a streamlined manner, there is a pressing need to distinguish between human and AI-generated actions in an online fashion.\footnote{ As an example of increasingly autonomous agent behavior in online environments, see the discussion at \url{https://www.reddit.com/r/singularity/comments/1r3fy5s/ai_agent_melts_down_after_github_rejection_calls/}.}
Second, statistical power of p-value methods may deteriorate when the NTP distribution becomes highly concentrated (i.e., nearly degenerate), a phenomenon that may occur in long text generation. Third, theoretical characterizations of power remain limited in general watermark detection frameworks, and existing methods often lack robustness against adversaries who might attempt to corrupt only later portions of a text.

To address these issues, we develop a unified framework for watermark detection based on e-values and e-processes, which are recently developed tools for online testing with useful features
\citep{GrunwaldHK19,  vovk2021values, wang2022false, ramdas2023game}; see
\cite{RW24} for background and a comprehensive treatment of the topic. Our framework reformulates watermark detection as sequential hypothesis testing of independence and constructs e-processes tailored to the watermark-induced structure, with a special focus on the popular Gumbel-max watermark \citep{aar2023}. The proposed approach enjoys the following key advantages:
\begin{enumerate}
    \item  The proposed testing procedures control Type I error under arbitrary stopping times, making them suitable for real-time and streaming detection scenarios.  Our method requires rather weak assumptions, and works well in the setting where a proportion of the text is human-edited.  
    \item Under mild assumptions, the proposed methodology gives rise to the only class of  admissible and unbiased e-processes, hence also the only class of  admissible and unbiased sequential tests. The connection between admissible e-processes and admissible sequential tests was established by \cite{ramdas2020admissible}.
    \item  By using  adaptive weights and the online Grenander algorithm \citep{hore2026online} to construct calibrators,  the proposed e-processes 
    are flexible and have good performance. In some situations, they can even have better performance than the existing methods that do not have sequential validity.  
    \item  We establish asymptotic power-one results under specific conditions, providing formal justification for the effectiveness of the proposed  procedures. 
\end{enumerate}

Together, these contributions provide a principled and theoretically grounded approach to watermark detection. Our framework unifies existing ideas under the lens of e-processes, clarifies the statistical structure of watermark detection, and offers both finite-sample guarantees and asymptotic efficiency. The anytime validity allows the detector to continuously monitor streaming outputs and stop as soon as sufficient evidence accumulates, thereby minimizing latency and computational cost without compromising statistical rigor. 

Although our study is motivated by the LLM watermark testing problem, the proposed methods and results are generally applicable to any online testing problems where independent pivotal statistics are sequentially available. This will be clear from our assumptions, most of which only concern the distributions of the test statistics under the null and the alternative hypothesis.

This article is organized as follows. Section \ref{sec:pre} introduces preliminary notions, including the problem setting and an overview of e-processes. In Section \ref{sec:formulation}, we rigorously formulate the watermark detection problem and derive the information bound measuring the difference between the null and alternative distributions, which applies to any watermark scheme. Building upon this foundation, Section \ref{sec:e-test} establishes a unified framework for watermark detection via e-processes with detailed construction strategies. In Section \ref{sec:guarantee}, we provide theoretical characterizations of the proposed procedures, focusing on their power properties. Empirical evaluations on both simulated data and open-source large language models are presented in Section \ref{sec:num} to demonstrate the practical advantages of our framework. Finally, Section \ref{sec:conclude} concludes the paper and discusses potential avenues for future research. The proofs of our theoretical results are provided in the Appendix.

\section{Preliminaries}
\label{sec:pre}
\subsection{Problem setting and notations}

Suppose that there are $K\ge 2$ tokens in the vocabulary $\mathcal W$ of the LLM.
We write $[n]=\{1,\dots,n\}$ for $n\in \N$ and $\mathcal W=[K]$.

 The statistical task is to test whether a given text with $T$ tokens  is generated by human, represented by a probability $\p$, or an LLM, represented by a probability $\q$. 
 As usual in online testing, the value of $T$ does not need to be known a priori, and it can be infinite.
 We set $\TP= [T]$ if $T<\infty$ and $\TP=\N$ if $T=\infty$. Moreover, $\T=\TP\cup\{0\}$.

Let $\Delta_K$ be the standard simplex in $\R^K$.  
For $\delta>0$, let $ \Delta^\delta_K$ be the set of  vectors $\bP$ in $   \Delta_K$ such that $\max (\bP)\le 1-\delta$. Implicitly we require $1-\delta\ge 1/K$; otherwise $ \Delta^\delta_K$ is empty, but this is not a problem as we are mainly interested in the case where $\delta$ is close to $0$.
Let $\Delta^\circ_K=\bigcup_{\delta>0}\Delta^\delta_K $, which is  the set of   vectors in $\Delta_K$ with at least two positive components. 
If $\ntp$ only has one positive component,  then we say that it is \emph{degenerate}.

For the $t$-th token, the NTP
is represented by a vector of probability $\ntp_t=(P_{t,w})_{w\in \mathcal W}$ in  $\Delta_K$.
We use a filtration $\mathcal F=(\mathcal F_{t})_{t\in \T}$ to represent the information available before generating each token,  and $\p_t$ and $\q_t$ are the conditional probability given  $\mathcal F_{t-1}$. We assume that the NTP vector $\ntp_t$ is $\mathcal F_{t-1}$-measurable for $t\in \TP$. 
 The vector $\mathbf P_t$ governs how the token $W_t$ is generated, that is,
 $$
\p_t(W_t=w) =\q_t(W_t=w) 
=P_{t,w} \mbox{ for $w\in \mathcal W$.}
$$

 For   $\mathbf P\in \Delta_K$, we write $\ntp=(P_1,\dots,P_K)$, and the entropy of $\mathbf P\in \Delta_K$ is defined as
 $$\ent(\ntp)=  -\sum_{w\in \mathcal W} P_w \log P_w .$$

We say that a random variable $Y$ is \emph{super-uniform} under a probability measure $\mathbb H$ if $\mathbb H(Y\le y)\le y$ 
for $y\in (0,1)$,
and it is \emph{strictly super-uniform} if $\mathbb H(Y\le y)< y$ 
for $y\in (0,1)$.

Throughout the paper, terms like ``increasing'' and ``decreasing'' are in the weak sense, and we follow the convention   $\sum_{\varnothing}=0$
and $\prod_{\varnothing}=1$.

\subsection{E-values and e-processes}
\label{sec:22}

For a null hypothesis, represented by a set of probabilities (a generic member is denoted by $\p$)  and a filtration $\mathcal F$, 
an e-variable is a nonnegative random variable $E$ with $\E^\p[E]\le 1$ (for all $\p$ in the null hypothesis), and 
an \emph{e-process} is a nonnegative stochastic process $(M_t)_{t \in \T}$  adapted to $\mathcal{F}$  such that $\mathbb{E}^\p [M_\tau] \leq 1$ for any stopping time $\tau$ with respect to $\mathcal{F}$. This means one can monitor evidence continuously without inflating error rates.

For an e-process $M$ with infinite time horizon, its asymptotic growth rate is defined as
the limit 
$$\liminf_{T\to\infty} \frac 1T \log M_T,$$
under a specific probability measure $\q$ in the alternative hypothesis.

Once an e-process $M$ is constructed, we can sequentially test  $\p$ versus $\q$ in the following way: 
reject the null hypothesis 
whenever the process $M$ reaches a threshold $1/\alpha$, where $\alpha \in(0,1)$ is a pre-specified level.
Ville's inequality (which is equivalent to the optional stopping theorem) gives
$$
\p\left(\sup_{t\in \T} M_t \ge \frac{1}{\beta }\right) \le \beta \mbox{ for all $\beta > 0 $}.
$$
Therefore, the type-I error is controlled, even if $T=\infty$, and we can use any stopping rule to terminate the process (e.g., stop after 1000 tokens, or stop after seeing $M$ drop below $0.1$).

 Nonnegative random variables $E_t$, $t\in \TP$, are \emph{sequential e-variables} if they satisfy
 $\E^\p[E_t|\mathcal F_{t-1}]\le 1$ for $t\in \TP$.
 An e-process $M$ is often constructed from the running product of sequential e-variables, that is,  
 $$
 M_t=\prod_{s=1}^t E_s,~~~t\in \T.
 $$
Note that 
  $ \log M_t  $ is the sum of $t$   terms $\log E_1,\dots,\log E_t$, 
and therefore, intuitively, $\log M_t$ would be large for large $t$ if the above terms  have positive means.
Our main task is to construct the e-values $(E_t)_{t\in \TP}$  such that the above intuition works rigorously. 

The use of e-values and e-processes in the setting of watermark detection naturally analyzes tokens one by one. That is, each e-value $E_t$ is constructed by one token at step $t$, so that they are actually independent. Then, we combine them into an e-process and  test  the hypothesis via Ville's inequality  as described above. 
For this reason, in the next section we first focus on the analysis of one token.

\section{A mathematical formulation of watermarks}
\label{sec:formulation}

\subsection{Watermarks and the detection problem}
For now, we focus on just one observation $(W,\zeta)$, where $\zeta$ is the pseudo-random vector that generates the next token, taking value in a set $Z$, and $W$ is the token that is chosen.
This  one-observation  setting is  without loss of generality  in constructing e-processes, as we will see later.

  The exact distribution of $\zeta$ does not matter, and we can without loss of generality assume that $\zeta$ is uniformly distributed on $Z=[0,1]^K$. Here $K$ is chosen as the size of the vocabulary $\mathcal W$, but again this choice is not important.

An (unbiased) watermarking scheme is a function
$S:  \Delta_K\times  Z
\to [K]
$
satisfying
\begin{align}
\label{eq:cond}
\q(S(\ntp, \zeta)=w)=P_w ,~~~\mbox{ for all } \ntp \in \Delta_K,~w\in \mathcal W,
\end{align}
and $\q$ is an objective data generating probability \citep{aos2025wm}. 
The condition in \eqref{eq:cond} can be rewritten as $S(\ntp, \zeta)\lawis \mathrm{Cat}(\ntp)$ for each $\ntp\in \Delta_K$, where $\mathrm{Cat} (\ntp)$ is the categorical distribution with probability vector $\ntp$.
There is no other requirement on $S$.

For one observed token $W\lawis \mathrm{Cat}(\ntp) $ and $\zeta$, the null hypothesis $H_0$
and the alternative hypothesis $H_1$ are given by 
$$
H_0: (W,\zeta) \mbox{ is  
independent; ~~~~ $H_1: W= S(\ntp,\zeta)$}. 
$$ 
For a given $\ntp$, 
we use  $\uu^\ntp $
for the distribution of  $(W,\zeta)$ under $H_0$
and 
$\q^{\ntp}$ for the distribution of  $(W,\zeta)$  under  $H_1$, to emphasize  that both $\uu^\ntp$ and $\q^\ntp$ depend on the usually unobservable $\ntp$. This makes the testing problem 
one with a
composite null and  a composite alternative. 
Moreover, $\q^\ntp$ depends on the watermarking scheme $S$ through  the condition \eqref{eq:cond} on $\q^\ntp$. 

\subsection{The Gumbel-max watermark}
\label{sec:Gumbel}

We next describe a primary example, the Gumbel-max watermark,  proposed by \citet{aar2023}. Our methodology does 
not rely on this specific watermark design, but it does not hurt to keep this example in mind for intuition.

As one of the most influential watermark schemes, it has been implemented internally at OpenAI \citep{openai2024understanding} and has become a baseline in many studies \citep{aos2025wm,li2025robust}. This watermark utilizes the Gumbel-max trick \citep{gumbel1954statistical,jang2017categorical} to ensure the unbiasedness characterized in \eqref{eq:cond}. The Gumbel-max watermark operates with the pseudo-random vector $\zeta = (U_{w})_{w \in \mathcal{W}}$ consisting of iid standard uniform variables, and the decoder is given by
\begin{equation}\label{eq:it}
S(\bP, \zeta) = \argmax_{w \in \mathcal{W}} \frac{\log U_w}{P_w}.
\end{equation}
As suggested by \citet{aar2023},
a pivotal statistic for one token is 
 $
 Y=U_{W}, 
 $
which is uniformly distributed on $[0,1]$ under the null hypothesis. Therefore, the distribution of $Y$ under $\uu^\ntp$ does not depend on $\ntp$, although the distribution of $(W,\zeta)$ depends on $\ntp$. Under the alternative hypothesis $\q^{\ntp}$, the distribution of $Y$   depends on $\ntp$, and it is  given by 
\begin{equation}
\label{eq:FPY}
    F^\ntp(y)= \q^\ntp (Y\le y)= \sum_{w\in \mathcal W} P_{w} y^{1/P_{w}} \mbox{~~~for $y\in [0,1]$, }
\end{equation} 
with the convention $y^{\infty}=0$ for $y<1$ and $1^\infty=1$.
Its density function is 
\begin{equation}
\label{eq:FPY-den}
 \frac{\d F^{\ntp}}{\d y}(y) = \sum_{w\in \mathcal W} y^{1/P_{w}-1}, ~~~~y\in [0,1], 
\end{equation}
 which is an increasing function of $y$.
 Moreover, it is strictly increasing on $[0,1]$ unless $\ntp$ is degenerate.   
 As a consequence, $Y$ is strictly super-uniform under the alternative hypothesis for any nondegenerate $\ntp $. Super-uniformity is a key property for our test based on e-processes, and motivates our main assumption in Section \ref{sec:e-test}.

Putting the subscript $t$ back to the original detection problem with $T<\infty$ tokens,  we denote by $Y_t=U_{t,W_t}$ the pivotal statistic for the $t$-th token $W_t$. \citet{aar2023} proposed to declare the presence of the watermark if the
sum-based statistic $H_{\mathrm{ars}} = \sum_{t=1}^T h_{\mathrm{ars}}(Y_t)$, with the score function $h_{\mathrm{ars}}(y) = -\log(1-y)$, is above a certain threshold. 
Other score functions, such as the log function $h_{\mathrm{log}}(y) = \log y$ \citep{fernandez2023wm,kuditipudi2024robust}, and the optimized score function $h_{\mathrm{gum}}^*$ analyzed in \cite{aos2025wm}, can also be adopted in the sum-based approach.

 The following lemma gives some   properties of  $F^\ntp(y)$ as a function of $\ntp$. 
For two vectors $x$ and $y$ in $\R^n$, we say that $y$ \emph{majorizes} $x$, denoted by $x \preccurlyeq y$, 
if $\sum_{i=1}^k x_{[i]}\le\sum_{i=1}^k y_{[i]}$ for all $k\in [n]$ and $\sum_{i=1}^n x_{[i]}=\sum_{i=1}^n y_{[i]}$, where $x_{[i]}$ is the $i$-th largest component of $x$ (and analogously for $y$).
For a set $A\subseteq \R^n$, a function $f: A\to \R$ is \emph{Schur-convex} if 
 $$
 x\preccurlyeq y ~\Longrightarrow~ f(x) \le f(y)
 $$
 for $x,y\in A$. See \cite{marshall2011inequalities} for a treatment of majorization and Schur-convexity. 
 
\begin{lemma}
\label{lem:gm}
Let $F^\ntp$ be given in \eqref{eq:FPY} and $\delta\in (0,1/2]$. 
\begin{enumerate}[(i)]
    \item    For $y\in [0,1]$, the function 
     $\ntp \to F^\ntp(y)$ is Schur-convex on $\Delta_K$. 
    \item   For  $y\in [0,1]$, the vector $\ntp^*=(1-\delta,\delta,0,\dots,0)\in \Delta_K^\delta$ 
     maximizes $F^{\ntp}(y)$ over $\ntp\in \Delta_K^\delta$. 
    \item For any increasing function $f:[0,1]\to \R$, we have 
     $$
    \min_{\ntp \in \Delta_K^\delta} \E^{\q^{\ntp}} [f(Y) ]
 = 
    \E^{\q^{\ntp^*}} [f(Y) ].
     $$
\end{enumerate}
\end{lemma} 
The properties in Lemma \ref{lem:gm} will be useful for our later analysis specific to the Gumbel-max watermark.

\subsection{Discussion on optimal e-values for watermarks}

We discuss how to construct an e-value based on one data point $(W,\zeta)$ and obtain some information bounds in this static setting, without specializing to any specific watermarking scheme. The main message is that standard approaches for optimal e-value, such as those described by \cite{GrunwaldHK19} and \cite{RW24}, depend crucially on the unknown NTP vector $\ntp$ and thus the corresponding tests are not suitable in practice. Nevertheless, the discussion provides some intuition for designing tests in Section \ref{sec:e-test}.

For the static setting, an e-variable is a statistic $E:\mathcal W\times Z\to \R_+$ satisfying the constraint $\E^\p[E(W,\zeta)]\le 1$. We write $E=E(W,\zeta)$ for simplicity. 
For a given $\ntp\in \Delta_K$, the first task is to consider

{
\begin{equation}
    \label{eq:opt-e} \begin{aligned}
\mbox{ maximize }  &  \E^{\q^{\ntp}}[\log E]
\\ \mbox{ over } &  \mbox{ statistic } E\ge 0; \mbox{ watermarking scheme } S
\\
\mbox{ subject to } & \E^{\uu^{\mathbf  R}}[E]\le 1  \mbox{ for all $\mathbf R\in \Delta_K $}.   
\end{aligned}
\end{equation}  
}

The quantity  $\E^{\q^{\ntp}}[\log E]$ is known as the e-power of $E$ against $\q^{\ntp}$.
There are two challenges associated with \eqref{eq:opt-e}, making it impossible to directly solve  in practice. First, The NTP vector $\mathbf P$ 
    is usually unavailable to the tester in a typical modern LLM. Second, the optimization problem \eqref{eq:opt-e} is very complicated, and we are not aware of any methods to solve it, due to the complexity caused by the watermarking scheme $S$.  

A natural upper bound on the e-power in \eqref{eq:opt-e}
is the Kullback--Leibler (KL) divergence of $\q^\ntp$
from $\uu^\ntp$, which measures the difference between the two measures. 
The KL divergence of a probability measure $\q_1$
from another probability $\q_2$ is defined  as \begin{align*}
\mathrm{KL}(\q_1 \Vert \q_2) 
&= \E^{\q_1 }\left[ 
\log\frac{\d \q_1 } {\d  \q_2} 
\right]  \mbox{~~~ if $\q_1  \ll \q_2 $,}
\end{align*}
 and it is $\infty$ otherwise. 
By \citet[Theorem 3.20]{RW24}, 
$$
\E^{\q^{\ntp}} [\log E] \le \mathrm{KL}(\q^{\ntp}\Vert \uu^\ntp) 
$$ for any e-variable $E$, 
and therefore $\mathrm{KL}(\q^{\ntp}\Vert \uu^\ntp) $ is an upper bound for the e-power in \eqref{eq:opt-e}. 
 The next proposition shows that $\mathrm{KL}(\q^{\ntp}\Vert \uu^\ntp) $ does not depend on the watermarking scheme.
\begin{proposition}
\label{prop:KLent}
For any watermarking scheme $S$ and  NTP vector $\ntp\in \Delta_K$, 
$$
\mathrm{KL}(\q^{\ntp}\Vert \uu^\ntp) = \ent(\ntp)
$$
always holds true. In particular, $\q^\ntp\ll \uu^\ntp$. 
\end{proposition}


In order for 
$ \E^{\q^{\ntp}} [\log E] =  \mathrm{KL}(\q^{\ntp}\Vert \uu^\ntp), $ 
it is necessary and sufficient that 
$E$ is given by 
$E= \d \q ^\ntp/ \d \uu^\ntp $ (\citealp[Theorem 3.20]{RW24}). However, this is usually not a feasible choice of an e-variable, because the constraint  in  \eqref{eq:opt-e} 
is guaranteed to satisfy only for the specific $\ntp$. 
A strategy to get $E$ satisfying  the constraint is 
to select a pivotal statistic, as in the case of the Gumbel-max watermark in Section \ref{sec:Gumbel}. For this choice, the maximum e-power for an e-variable is given by
$$\mathrm{KL}(F^\ntp\Vert \mathrm{U})= \sum_{w\in \mathcal W}  \int_0^1 y^{1/P_{w}-1} \log \left(\sum_{w\in \mathcal W}   y^{1/P_{w}-1} \right)\d y \le \ent(\ntp) ,$$
where $\mathrm U$ is the uniform distribution on $[0,1]$. 
Since the mapping from $\zeta$ to $Y$ is not injective, typically
$\mathrm{KL}(F^\ntp\Vert \mathrm{U})< \ent(\ntp)$ holds.  Hence, by using the Gumbel-max watermark,  the e-power has to be reduced, even if $\ntp$ were known. 
In the next section, we take the approach of   pivotal statistics, and design e-values that do not depend on $\ntp$.


\section{E-processes for  watermarks}
\label{sec:e-test}

Throughout this section, we let
 $Y_t,~t\in \TP$ be  the pivotal statistics for  a  watermark scheme.  
Recall that
$\mathcal F=(\mathcal F_t)_{t\in \T}$ is the information filtration, to which $(Y_t)_{t\in \TP}$ is adapted, with the interpretation that
 $\mathcal{F}_{t-1}$  contains all the statistical information available before generating the token $w_t$.  
 The sequence $(Y_t)_{t\in \TP}$ satisfies two basic conditions.
 
\begin{assumption}
     \label{assm:Y}
The distributions of $Y_t$ satisfy that:
{
\begin{enumerate}[(i)]
\item under the null hypothesis, $Y_t$ is uniformly distributed on $[0,1]$ and independent of $\mathcal F_{t-1}$;
\item under the alternative hypothesis, $Y_t$ is super-uniform conditional on $\mathcal F_{t-1}$.
\end{enumerate}}
\end{assumption}

 The uniform distribution in Assumption \ref{assm:Y}  is not essential, because 
continuously distributed pivotal statistic can be converted to a uniformly distributed one.
Condition (ii) does not matter for the validity (type-I error control) of our method, but it affects the power of the method.

Assumption \ref{assm:Y} 
allows for the case where some of the texts are generated by human edits under the alternative hypothesis, a setting studied by \cite{li2025robust}.  In this setting, some of $Y_t$ (LLM generated) are super-uniform (and not uniform), 
whereas  the others (human generated or edited) are uniform; see the experiments in Section \ref{sec:num}.

In the simple case where
$\mathcal F_t=\sigma\{Y_1,\dots,Y_t\}$ and $Y_t$,  $t\in \TP$ are independent,
(i) means that  $Y_t,~ t\in \TP$ are iid uniformly distributed on $[0,1]$.
 The special case of the Gumbel-max watermark is explained in Section \ref{sec:Gumbel}, where 
$Y_t=U_{t,w_t}$ for $t\in \TP$  and  they satisfy Assumption \ref{assm:Y}.

\subsection{Online test with e-processes}

As explained in Section \ref{sec:formulation}, there is no feasible way to directly construct e-values and e-processes to optimize their power  against the alternative hypothesis without knowing the NTP. Here we present a methodology, which, although not optimal in the sense of \eqref{eq:opt-e}, essentially covers  all useful e-processes for the testing problem in Assumption \ref{assm:Y}. 

We first introduce (p-to-e)  calibrators  \citep{vovk2021values}, which convert p-values into e-values. A \textit{calibrator}
is a decreasing function $f:[0,1]\to\R_+$ satisfying 
$\int_0^1 f(p)\d p=1$.
Although in some literature only $\int_0^1 f(p)\d p\le 1$ is required, it is without loss of generality to   consider only those  with $
\int_0^1 f(p)\d p=  1
$ because a large e-value is better for the test.

To use calibrators in an online testing framework, we need random calibrators. 
 For $t\in \TP$,
 let $\Phi_t$ be the set of all random functions $f_t:[0,1]\to \R_+$ such that 
\begin{enumerate}[(a)]
    \item $f_t$ is determined by $Y_1,\dots,Y_{t-1}$ (that is, $f_t$ is predictable);
    \item  $f_t$ is a calibrator for all realizations of $Y_1,\dots,Y_{t-1}$.
\end{enumerate}
In other words, $\Phi_t$ is the set of  all predictable calibrators at time $t$.   
We consider the e-process $M=(M_t)_{t\in \T}$ defined by  
\begin{equation}
    \label{eq:M}M_t = \prod_{s=1}^t   E_s, ~~~\mbox{ for } t\in \T,
\end{equation}
where each $E_t=f_t(1-Y_t)$ for some  $f_t \in \Phi_t$.

Adapting the standard terminology in hypothesis testing, 
we say that an e-process $M$
is \emph{unbiased} (against the alternative hypothesis) 
if for all probability measures $\q$ in the alternative hypothesis,  $\E^\q[M_t]\ge 1$
for all $t\in \TP$. 
In other words, its mean is at most $1$ under the null hypothesis,
and at least $1$ under the alternative hypothesis.


The next  result, which is the most important technical tool for constructing our tests,   follows from standard arguments in the e-value literature.  
 \begin{theorem}
 \label{th:validity}
 Under Assumption \ref{assm:Y}, the process $M=(M_t)_{t\in \T}$  in \eqref{eq:M}  with $$
E_t=f_t(1-Y_t),~~~t\in \TP$$  
is an unbiased e-process (also a martingale), where $f_t\in \Phi_t$ for each $t$.
In particular, under the null hypothesis, 
$$
\p\left(\sup_{t\in \T} M_t\ge 1/\alpha\right) \le \alpha 
~~~\mbox{ for all $\alpha \in (0,1)$.}
$$
 \end{theorem}

The next result gives a converse statement to Theorem \ref{th:validity}, showing that \eqref{eq:M} is the only possible form for e-processes when $\mathcal F$ is the natural filtration of  $(Y_t)_{t\in \TP}$. 
An e-process $M$ is \emph{admissible} if there is no other e-process $M'$ such that $M'\ge M$ with $\p(M'_t>M_t)>0$ for some $t\in \TP$.
\begin{theorem}\label{th:revese}
Suppose that $\mathcal F_t=\sigma(Y_1,\dots,Y_t)$ for $t\in \TP$ and Assumption \ref{assm:Y}  holds.  
Then $M$ is an admissible and unbiased e-process    if and only if
\begin{equation}\label{eq:dominance}
M_t=  \prod_{s=1}^t f_s(1-Y_s),~~~t\in \T,~ \mbox{~$\p$-almost surely}.
\end{equation} for some $f_t \in \Phi_t$ for all $t\in\TP$.
\end{theorem}

As a further step, \citet[Theorem 18]{ramdas2020admissible}  also implies that all admissible sequential test $(\psi_t)_{t\in \TP}$ at level $\alpha$ for the null hypothesis must  have the form 
$$\psi_t=\max_{s\le t}\id_{\{M_s\ge 1/\alpha\}},~~~t\in \TP$$ for some e-process $M$, that is,  rejecting the null hypothesis when   $M$ reaches $1/\alpha$. Therefore, from Theorems \ref{th:validity} and \ref{th:revese}, we conclude that the method we proposed in Algorithm \ref{alg:sebh} is essentially the only way to build sequential tests for this problem.

We summarize our online watermark detection procedure in Algorithm \ref{alg:sebh},
where the null hypothesis is rejected as soon as $M$ reaches $1/\alpha$ or above. 
Two parameters $T\in \N\cup\{\infty\}$  and $\beta\in [0,1)$ can be included optionally to terminate the test, declaring no rejection of the null hypothesis. 
More precisely,
the online test stops when $T$ data points are exhausted 
or the e-process falls below $\beta$.
The default values are $T=\infty$ and $\beta=0$, meaning that there is no early termination. 


\begin{algorithm}[t]
\caption{Online watermark detection using an e-process}  
\label{alg:sebh}
\begin{algorithmic}[1]  
\REQUIRE The sequence of data points $Y_t$, a maximum number $T\in \N\cup\{\infty\}$ of tokens, a termination threshold $\beta\in [0,1)$, and a nominal type I error level $\alpha \in (0,1)$
\STATE Set $Y_0=0$ and $M_0 = 1$
\FOR{$t = 1,2,\dots, T$}
    \STATE Compute the calibrator $f_t\in \Phi_t$ based on $Y_1,\dots,Y_{t-1}$
    \STATE Read $Y_t$ and compute $M_t = M_{t-1} \times f_t(1 - Y_t)$
    \IF{$M_t \ge 1/\alpha$}
        \RETURN ``rejection'' and terminate
    \ENDIF
    \IF{$M_t < \beta$}
        \RETURN ``no rejection'' and terminate
    \ENDIF
\ENDFOR
\RETURN ``no rejection''

\end{algorithmic}
\end{algorithm}

In the next subsections, we present a few specific ways of  constructing e-processes in the proposed detection procedure.

\begin{remark}
It should be clear from Assumption \ref{assm:Y} and Theorem \ref{th:validity} that the proposed methodology is  applicable to any online testing problem where pivotal statistics are sequentially available, such that 
they have known continuous distributions under the null hypothesis. 
The uniqueness result in Theorem \ref{th:revese} further requires that their distributions become stochastically larger under the alternative hypothesis.  
The proposed tests 
also remain valid when the statistics have distributions that are stochastically smaller than known distributions.
\end{remark}

\subsection{Fixed calibrators with adaptive weights}
\label{sec:fixed}

We first consider some fixed choices of deterministic calibrators, which are nonnegative decreasing functions  on $[0,1]$ that integrate to at most $1$.  Some common examples \citep{vovk2021values} are, as functions of $p$,
$$    2(1-p),~~
      p^{-1/2}-1,~~
  -\log{p}, ~~
    \mbox{and} ~~
  \frac{1-p+p\log{p}}{p(\log{p})^2}.$$
  For practical use, we always assume   a calibrator $g$ satisfies 
$
\int_0^1 g(p)\d p =1.
$

For a  calibrator $g$, 
and $t \in \TP$,  let 
$$ E_t=(1-\lambda_t)+\lambda_t g(1-Y_t),$$
where $\lambda_t$ is $[0,1]$-valued and measurable with respect to $\mathcal F_{t-1}$. The e-process $M=(M_t)_{t\in \T}$ is then constructed by $M_t=\prod_{s=1}^t E_s$ as   in \eqref{eq:M}.
With a given $g$,  choosing $(\lambda_t)_{t\in \TP}$ is a nontrivial task and affects the power of the e-process. 
A simple choice is to use a pre-specified  $\lambda \in(0,1)$, resulting in 
the process
\begin{equation}\label{eq:fixed}
M_t=\prod_{s=1}^t (1-\lambda+\lambda g(1-Y_s)),~~~t \in \T.
\end{equation}
We will call the e-process in \eqref{eq:fixed} a \emph{nonadaptive e-process}.

For better power, $(\lambda_t)_{t\in \TP}$ should be chosen adaptively. 
Using the idea of the {empirically adaptive e-process} in \citet[Section 7.8]{RW24}, $\lambda_t$ is given by $\lambda_1=0$ and 
\begin{equation}\label{eq:adaptive} \lambda_t= \argmax_{\lambda\in[0,\gamma]} \left\{\sum_{s=1}^{t-1}
\log\left((1-\lambda)+ \lambda  g(1-Y_s)\right)
\right\},~~~t\ge 2,\end{equation}
  where $\gamma \in (0,1]$. The choice of $\gamma$ is not important as the maximizer in \eqref{eq:adaptive} is typically small except for the first few values of $t$, and we will use the typical choice $\gamma=1/2$ in all implementations.
For a calibrator $g$ and $(\lambda_t)_{t\in \TP}$ in \eqref{eq:adaptive}, we  write 
\begin{equation}\label{eq:tildef}\tilde f^{g}_t=(1-\lambda_t) + \lambda_t g.\end{equation}
It is clear that $\tilde f^{g}_t$  belongs to $\Phi_t$, and hence Theorem \ref{th:validity} applies to $\tilde f^{g}_t$. 
The associated e-process is thus
\begin{equation}\label{eq:EAE}
M_t=\prod_{s=1}^t \tilde f^{g}_t (1-Y_t),~t\in \T,~ \mbox{with any nonconstant $g$ and $\gamma\in (0,1)$},
\end{equation}
and we will call it a  \emph{weight-adaptive  e-process}.

In our numerical experiments, we find that  the calibrator $g(p)=-\log(p)$ with adaptively chosen $\lambda_t$ in \eqref{eq:adaptive} leads to superior performance over other calibrators or fixed $\lambda$. 

\subsection{Online Grenander e-processes}
\label{sec:OG}

Let $\mathcal D$ be the set of all decreasing density functions on $[0,1]$ (that is, the set of calibrators).  
Define the  online Grenander (OG) calibrator \citep{hore2026online}, by 
\begin{equation}\label{eq:EA}
f^{\mathrm{OG}}_{t}= \argmax_{f\in\mathcal D}\sum_{s=0}^{t-1}\log f(1-Y_s), ~~~t\in \TP,
\end{equation} 
where we set $Y_0=0$.
The computation of \eqref{eq:EA} is straightforward as it coincides with the classic   \cite{grenander1956theory}  estimator  applied to $(Y_0,Y_1,\dots,Y_{t-1})$: It is known that the unique maximizer in \eqref{eq:EA}  $f^{\mathrm{OG}}_{t}$, which is unique on $(0,1]$ and flexible at $0$, is the left derivative of the smallest concave majorant of the empirical distribution function of $Y_0,Y_1,\dots,Y_{t-1}$ (see e.g.,  Chapter 8 of \citealp{devroye1987course}). We say that it is unique because the value of $f_t^{\mathrm{OG}}$ at $0$ does not matter, and we can safely set $f_t^{\mathrm{OG}}$ to be continuous at $0$.
The function $f^{\mathrm{OG}}_{t}$  is a left-continuous step function. 
Let  $$
   E_t= f_{t}^{\mathrm{OG}}(1-Y_t),
~~~\mbox{$t\in \TP$.}
 $$
Clearly, $f_t^{\mathrm{OG}}\in \Phi_t$ for each $t$.  
Note that $f_1^{\mathrm{OG}}
=\argmax_{f\in\mathcal D}\log f(1),
$
which is the constant function $1$. Therefore, at the first step, the calibrator is chosen as the constant $1$, and it evolves as the data $Y_1,Y_2,\dots$ are processed.  
Note that if we remove the term $i=0$
from \eqref{eq:EA}, then we would have  
$f_t^{\mathrm{OG}}(y)=0$ for all $y>1-\min\{Y_1,\dots,Y_{t-1}\}$, yielding 
a positive probability for  
$E_t=f_t^{\mathrm{OG}}(Y_t)=0$ (which has the same probability as $Y_t<\min\{Y_1,\dots,Y_{t-1}\}$). 
This is a situation that we would like to avoid; note that $E_t=0$ implies $M_s=0$ for all $s>t$. 
Introducing
$Y_0=0$ into \eqref{eq:EA} addresses this issue, but it may be too conservative, as it acts as if the first p-value $1-Y_0$ is $1$.
A more lenient choice is 
\begin{equation}\label{eq:EA2}   \argmax_{f\in\mathcal D}\left\{\sum_{s=1}^{t-1}\log f(1-Y_s)+\frac{\log f(0)}{2} +\frac{\log f(1)}{2} \right\}, ~~~t\in \TP.
\end{equation}  
Similarly to the case of \eqref{eq:EA}, \eqref{eq:EA2}  also yields a constant calibrator $1$ at step $t=1$.
Intuitively, 
\eqref{eq:EA2} uses a Bernoulli distributed $Y_0$ with mean $1/2$,
and its computation   is the same as the Grenander estimator in \eqref{eq:EA}. 
The optimizers in each of \eqref{eq:EA}--\eqref{eq:EA2}  belong to $\Phi_t$, and hence Theorem \ref{th:validity} applies to them. 
We will call 
the associated e-process 
\begin{equation}
    \label{eq:OGE}
        M_t= \prod_{s=1}^t f_s^{\rm OG}(1-Y_s),~t\in \T, \mbox{ with $(f^{\rm OG}_s)_{s\in \T}$ from   \eqref{eq:EA} or \eqref{eq:EA2}}
\end{equation} 
an \emph{OG e-process}.
Here, for different $s,t$, we allow $f_s^{\rm OG}$  from \eqref{eq:EA}  and 
$f_t^{\rm OG}$ from \eqref{eq:EA2}, although practically there is no benefit to do so. 

Finally, we can include a range  $[a,b]$ with $0<a<1<b<\infty$ to the  density functions in $\mathcal D$ in \eqref{eq:EA}--\eqref{eq:EA2}, and we call the resulting e-process in \eqref{eq:OGE} the \emph{OG e-process with range $[a,b]$}. Note  $f_t^{\rm OG}$  in \eqref{eq:EA} has a natural floor $f_t^{\rm OG}\ge 1/t$ due to $Y_0=0$, which is sufficient for practical purposes. We need the fixed range $[a,b]$ only for a theoretical result in Section \ref{sec:guarantee}.

\begin{remark}
A variant of the above construction is, instead of using all data $s=1,\dots,t-1$ up to step $t-1$ to compute $f^{\rm OG}_t$, 
to use a moving window such as $s=t-k,\dots,t-1$ for a fixed $k$,
which may work better for some time-series models. For the problem of LLM watermark detection in this paper, such a variant does not seem to be beneficial. 
Moreover, instead of a Bernoulli distribution for $Y_0$, we can choose any distribution with mean no larger than $1/2$ and support including $0$, which achieves the same purpose.
\end{remark} 

In our experiments, we use the formulation in \eqref{eq:EA2}, which avoids the issue of being overly conservative and it is computationally efficient.

The OG calibrator is useful in several settings. For instance, \cite{hore2026online} showed that in an iid setting, the OG e-process with range $[a,b]\subseteq (0,\infty)$ has an optimal regret bound.

\subsection{Average e-processes}
\label{sec:average}
Our practical recommendation for constructing the online test is to use the arithmetic average of the OG e-process and the weight-adaptive e-process with the log calibrator $g=-\log$. That is, to use
\begin{equation}
    \label{eq:average}
    M_t= \frac{1}{2} \prod_{s=1}^t \tilde f^{g}_s (1-Y_s) + \frac{1}{2}\prod_{s=1}^t f_s^{\rm OG}(1-Y_s),~~~t\in \T.
\end{equation}
Clearly, the average of e-processes is an e-process.
Alternatively, one can choose different weights than $(1/2,1/2)$ in \eqref{eq:average}. For simplicity, we will call the e-process in \eqref{eq:average} an \emph{average e-process}. 

The test based on the average e-process  in \eqref{eq:average} may have better performance than both tests based on the individual ones, as shown by our empirical results in Section \ref{sec:num}. 
As a feature of e-processes, the arithmetic average of two e-processes typically has the advantages of the two;
see e.g., the mixture e-processes  in \cite{RW24}. In particular, it is straightforward to see that the asymptotic growth rate of the average e-process is the maximum of those of the two individual e-processes.
Taking an average of more than two e-processes (e.g., with different calibrators) also has such a feature, but it does not yield noticeable further improvements on power performance according to  our empirical experiments.

 \section{Theoretical power guarantees}

\label{sec:guarantee}
In practice, the next-token prediction $\ntp_t$ is not available to the tester,  and it is not stationary. This makes theoretical power guarantees difficult to analyze.
 We will first consider a situation that is relevant for the token generating process. 
Throughout this section, to state asymptotic results, the time horizon is infinite,
that is, $\TP=\N$.
\begin{assumption}
\label{assm:randomness}
For some $\delta>0$ and $m\in \N$, 
 $\ntp_t\in \Delta^\delta_K$ with probability $1$ for each $t$, and $(\ntp_t)_{t \in \TP}$ is $m$-dependent (that is, $\sigma(\ntp_{t}:t\le t_1)$ and $\sigma(\ntp_t:t\ge t_2)$ are independent when $t_2-t_1>m$). 
\end{assumption}

Note that the distribution of $(\ntp_t)_{t \in \TP}$ is the same under both the null and the alternative hypotheses. 
The assumption $\delta>0$ means that the temperature of the LLM is bounded away from $0$,
and thus the token generating process has enough randomness. 
The assumption of $m$-dependence means that the generation of one token   only relies on a maximum of $m$ past tokens, and it is consistent with the real-world LLM generation mechanism \citep{aar2023,kir23llm,aos2025wm}.
This assumption is needed to apply a form of the law of large numbers.
Some other forms of weak or short-term dependence such as $\alpha$-mixing are also sufficient for this purpose, and we omit such discussions.

With an e-process $M$ on an infinite horizon and $\alpha\in(0,1)$, we say that the online test at level $\alpha$ is \emph{consistent} 
if
$$
\q\left(\sup_{1\le t\le T} M_t\ge \frac{1}{\alpha}\right) \to 1 \mbox{ as $T\to\infty$;}
$$
that is, the probability of rejecting the null hypothesis goes to $1$ if the data are generated under the alternative hypothesis. 
The next result is specific to the Gumbel-max watermarking scheme described in Section \ref{sec:Gumbel}.

\begin{theorem}\label{th:fixed}
Suppose that Assumption \ref{assm:randomness} holds and $g$ is a nonconstant calibrator. For the Gumbel-max watermark, there exists $\lambda \in(0,1)$ that depends only on $\delta$ and $g$, such that 
the nonadaptive e-process $M$  in \eqref{eq:fixed}  has exponential growth under the alternative hypothesis: 
$$
\liminf_{T\to\infty} \frac 1T \log M_T  >0  \mbox{~almost surely.}
$$
In particular, the  online test at any level $\alpha\in(0,1)$ is consistent.
\end{theorem}

A useful lemma will facilitate the proof of Theorem \ref{th:fixed}.

\begin{lemma} \label{lem:E1}
Let $Y$ be the   Gumbel-max statistic in Section \ref{sec:Gumbel}. 
    For any  $\delta>0$ and any nonconstant calibrator $g$,
    there exists $\lambda_0\in (0,1)$ such that   
    $$\inf_{\ntp\in\Delta^\delta_K} \E^{\q^\ntp} [\log (1-\lambda+\lambda g(1-Y))]>0$$
    for all $\lambda\in (0,\lambda_0]$. 
\end{lemma}

The key property of the Gumbel-max watermark used in Theorem \ref{th:fixed} is the uniform bound  over $\Delta^\delta_K$ established in Lemma \ref{lem:gm}. 
For given $g$ and $\delta$,
the value of $\lambda $ in Theorem \ref{th:fixed} can be computed from solving for any $\lambda_0>0$ satisfying \eqref{eq:lambda0}, which is a simple numerical task.

  Theorem \ref{th:fixed} gives a power-one guarantee for the nonadaptive e-process with a suitably chosen fixed $\lambda\in (0,1) $, 
which supports using the online test.
Nevertheless, the power  of the online test for nonadaptive e-process  may not be good, because this guarantee needs to guard against all possible NTP vectors, including the most adversarial, but unrealistic, ones. 
In numerical experiments, we find that 
the adaptive weights in \eqref{eq:tildef} give rise to much more powerful e-processes. 

Next, for a high-level understanding, we consider an idealistic situation. Under this assumption,  we can establish the 
consistency of the weight-adaptive e-process 
and 
the OG e-process.

\begin{assumption}
\label{assm:iid}
  The vectors $\ntp_1,\ntp_2,\dots$  are deterministic and identical to $\ntp\in \Delta_K^\circ$.
\end{assumption}
Although Assumption \ref{assm:iid} is quite   unrealistic,  it helps to highlight that 
the weight-adaptive e-process and the OG e-process are reasonable choices, at least in an idealistic setting. The next   result is not specific to the Gumbel-max watermark, and applies to all watermark schemes with statistics satisfying the following assumption.
\begin{assumption}
\label{assm:power}
Given the constant $\ntp$ in Assumption \ref{assm:iid}, $Y_t$, $t\in \TP$ are iid with a strictly super-uniform distribution $F^\ntp$ satisfying $\mathrm{KL}(F^\ntp \Vert \mathrm{U})<\infty$.
\end{assumption}
 It is straightforward to check that Assumption \ref{assm:power} is satisfied by the Gumbel-max watermark.

\begin{theorem}\label{th:OG}
Suppose that Assumptions  \ref{assm:iid} and \ref{assm:power} hold.
 The weight-adaptive e-process $M$   in \eqref{eq:EAE},  the 
 OG  e-process $M$ in \eqref{eq:OGE} with range $[a,b]$, and the average e-process in \eqref{eq:average} all have 
 exponential growth under the alternative hypothesis: 
$$
\lim_{T\to\infty} \frac 1T \log M_T  >0 \mbox{~in probability.}
$$
In particular, the online test at any level $\alpha\in(0,1)$ is consistent.
\end{theorem}

\section{Experiments on open-source LLMs}
\label{sec:num}

In this section, we investigate the performance of different methods in watermark detection. Specifically, we assess the power of detection procedures based on token sequences generated by simulated or open-source language models with the Gumbel-max watermark (Section \ref{sec:Gumbel}).
The comparative analysis includes the following detection methods:

{
\begin{itemize}
\item The sum-based method using the score function $h_{\rm ars}(y)=-\log(1-y)$ \citep{aar2023}.
\item The sum-based method using the score function $h_{\log}(y)=\log(y)$ \citep{fernandez2023wm,kuditipudi2024robust}.
\item The optimized sum-based methods using the score function $h_{\text{gum},d}^*$ proposed by \cite{aos2025wm}, with the hyper-parameter $d=0.1$ and $d=0.01$, respectively.
\item The weight-adaptive e-process in \eqref{eq:EAE}, with the fixed calibrator $g(p)=-\log(p)$.
\item The  OG  e-process in \eqref{eq:OGE}     with $f^{\rm OG}_t$ calculated by \eqref{eq:EA2}.
\item The average e-process  in \eqref{eq:average} with $g(p)=-\log(p)$.
\end{itemize}
}

We first evaluate these methods on simulated data. The experimental setup essentially follows the approach described in \citet{aos2025wm}. We generate a vocabulary $\mathcal{W}$ of size 1,000 and assess Type II, Type I and sequential Type I errors throughout 1,000 repeated experiments. 
Here, sequential Type I errors are computed through sequential monitoring; that is, for $t=1,2,\dots,$ we run the test on $t$ tokens, and declare rejection (watermark is detected) as soon as the test for $t$ tokens rejects. For traditional tests that are designed for fixed token length, this will inflate the Type I error. 

In our simulation, we model LLM’s NTP distributions as spike distributions: we set its largest probability as $1-\Delta_t$, where $\Delta_t$ is iid sampled from $\mathrm{U}(0.001,\delta)$ for $t=1,\dots,700$, and the remaining probabilities are firstly sampled from $\mathrm{U}(0,1)$ and subsequently normalized such that $||\mathbf{P}_t||_1=1$. We assess (sequential) Type I error rates on unwatermarked text, and Type II error rates on text generated from the same model with the Gumbel-max watermark.
The error rate curves are depicted in Figure \ref{fig0}.

We proceed to investigate the empirical performance of watermark detection methods for text sequences generated by the LLM OPT-1.3B, following the experimental setup in prior works \citep{kir23llm,aos2025wm}. Although OPT-1.3B has fewer parameters compared to commercial LLMs (and it may therefore be regarded as a small language model), evaluating watermark detection methods is a statistical task that is not about the quality or intelligence level of the text generated; we expect similar patterns for the power of watermark detection for more complicated LLMs. The average Type I and Type II error curves are evaluated based on 500 independent repeated experiments. We present the numerical results under two different temperature parameters (0.5 and 1, respectively) in Figure \ref{fig1}.

Finally, we implement these methods to detect watermarked text under human editing. After obtaining the initial pivotal statistics $(Y_t)_{t\in\TP}$, we let $Y'_t = (1-B_t)Y_t + B_t U_t$ for $t\in\TP$, where $U_t$ are iid $\mathrm{U}(0,1)$ variables, $B_t=0$ for $t\leq 50$ and $B_t$ are iid $\mathrm{Bin}(r)$ variables for $t>50$. Subsequently, we substitute $(Y'_t)_{t\in\TP}$ for $(Y_t)_{t\in\TP}$ as the input to all detection methods. This modification can effectively synthesize the scenario where a portion of the generated tokens is replaced by human-written content. The parameter $r$ indicates the proportion of human edits. We set $r\in \{0.1,0.3,0.5\}$ based on the setup used in the top row of Figure \ref{fig1} and evaluate the Type II error rates based on 500 independent repeated experiments. The results are depicted in Figure \ref{fig2}.

\begin{figure}[t]
   \centering
   \includegraphics[width=1\linewidth]{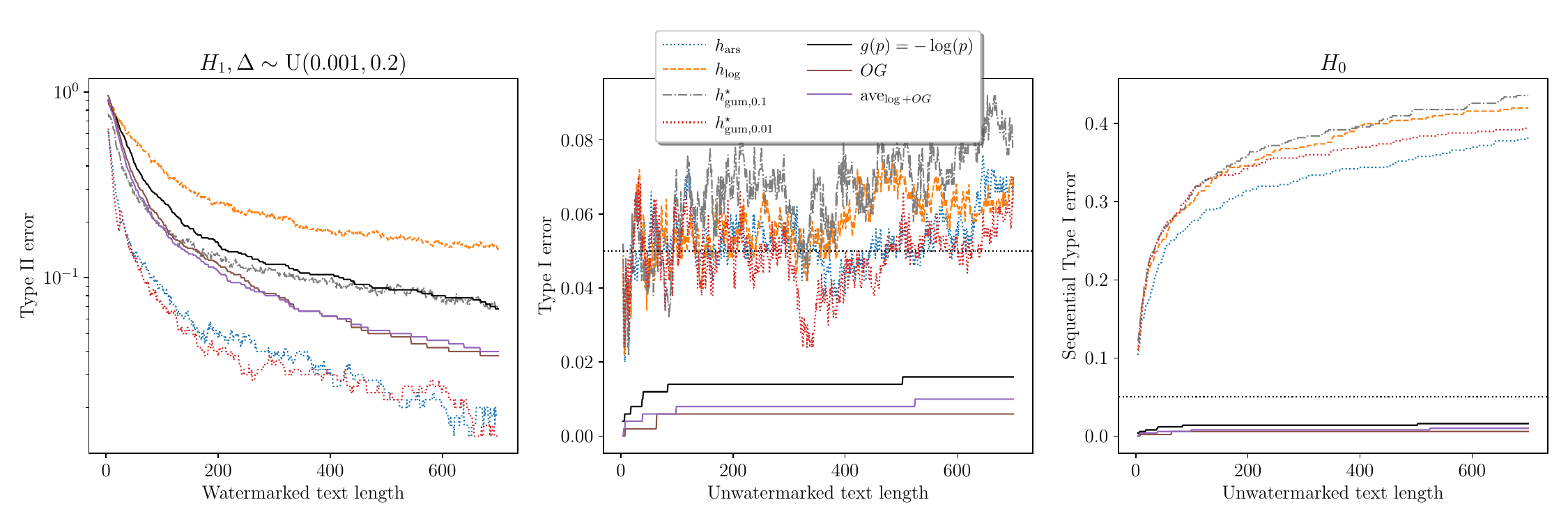}
   \includegraphics[width=1\linewidth]{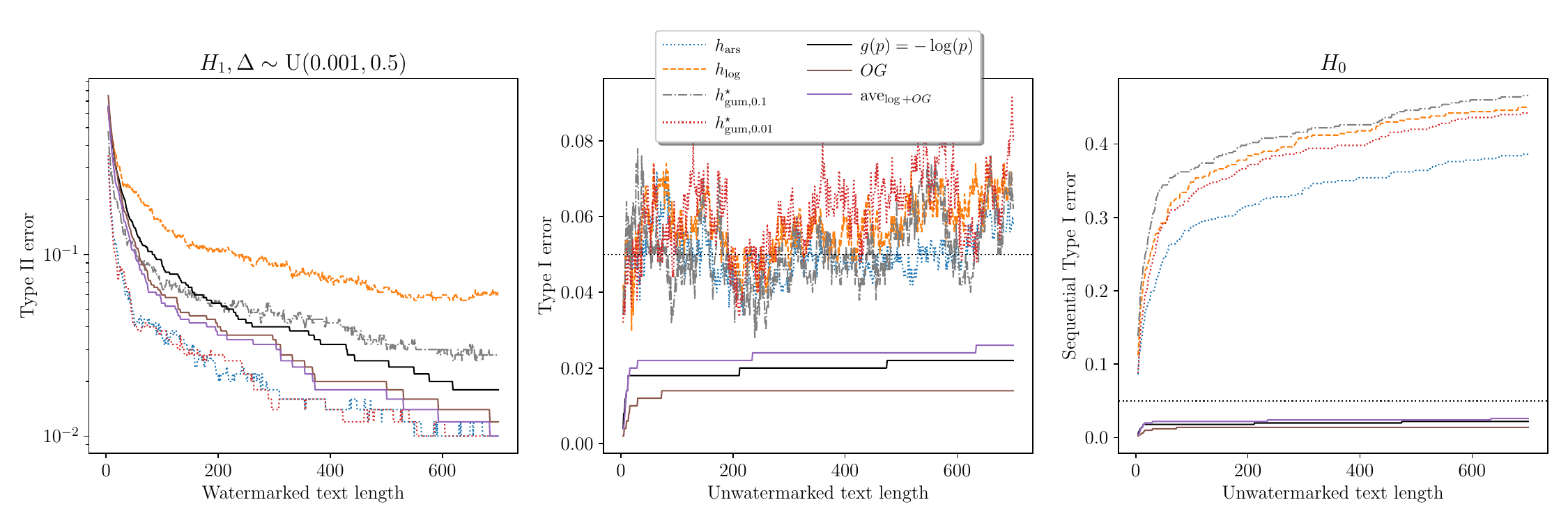}
   \caption{ Average Type II error (on a log scale), Type I error and Sequential Type I error rates versus text length on the simulated data with $\delta=0.2$ (top), and $\delta=0.5$ (bottom).}
   \label{fig0}
\end{figure}

\begin{figure}[t]
    \centering
\includegraphics[width=1\linewidth]{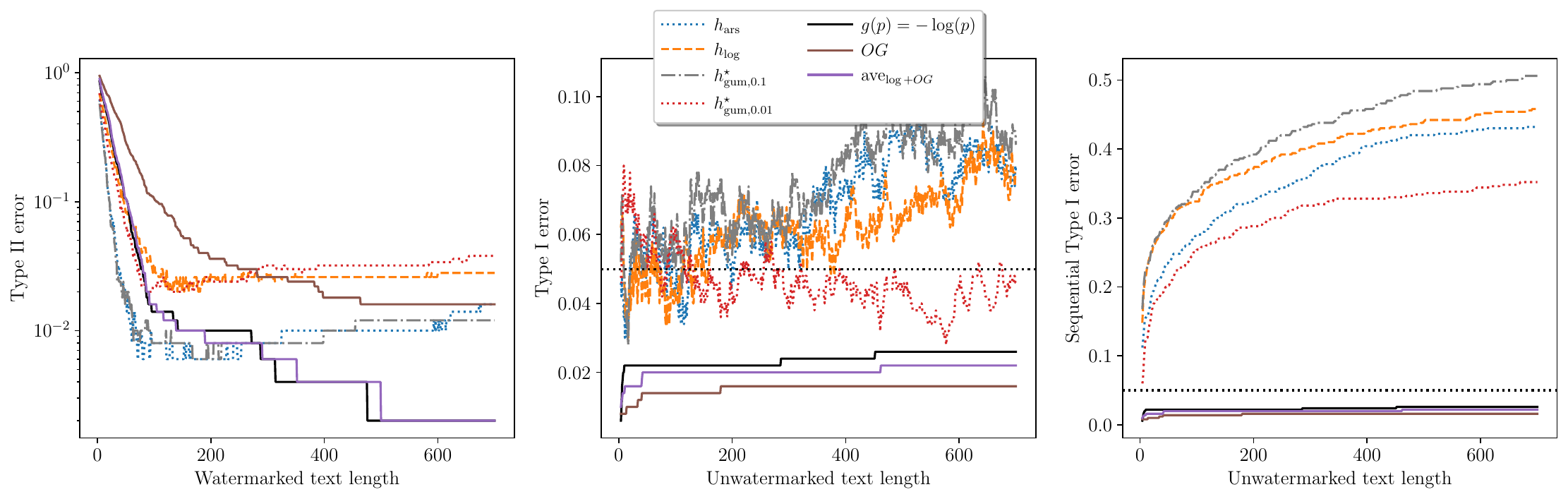}
\includegraphics[width=1\linewidth]{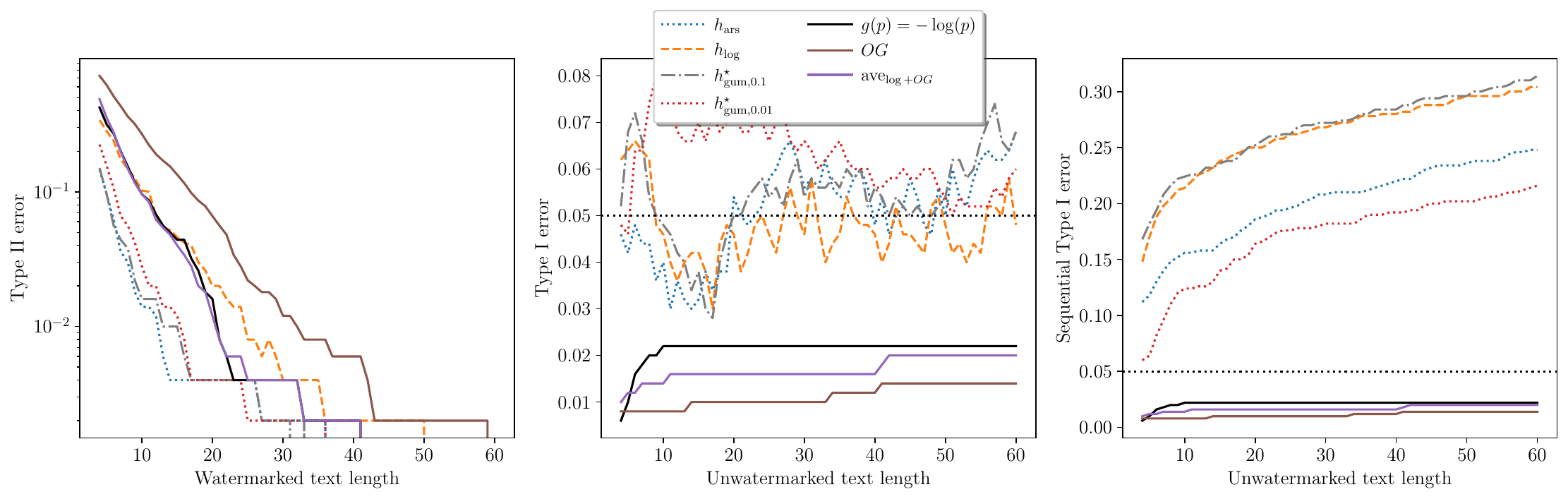}
    \caption{Average Type II error (on a log scale), Type I error and Sequential Type I error rates versus text length on the OPT-1.3B model with temperature parameters at 0.5 (top), and 1 (bottom).}
    \label{fig1}
\end{figure}

\begin{figure}[t]
    \centering
\includegraphics[width=1\linewidth]{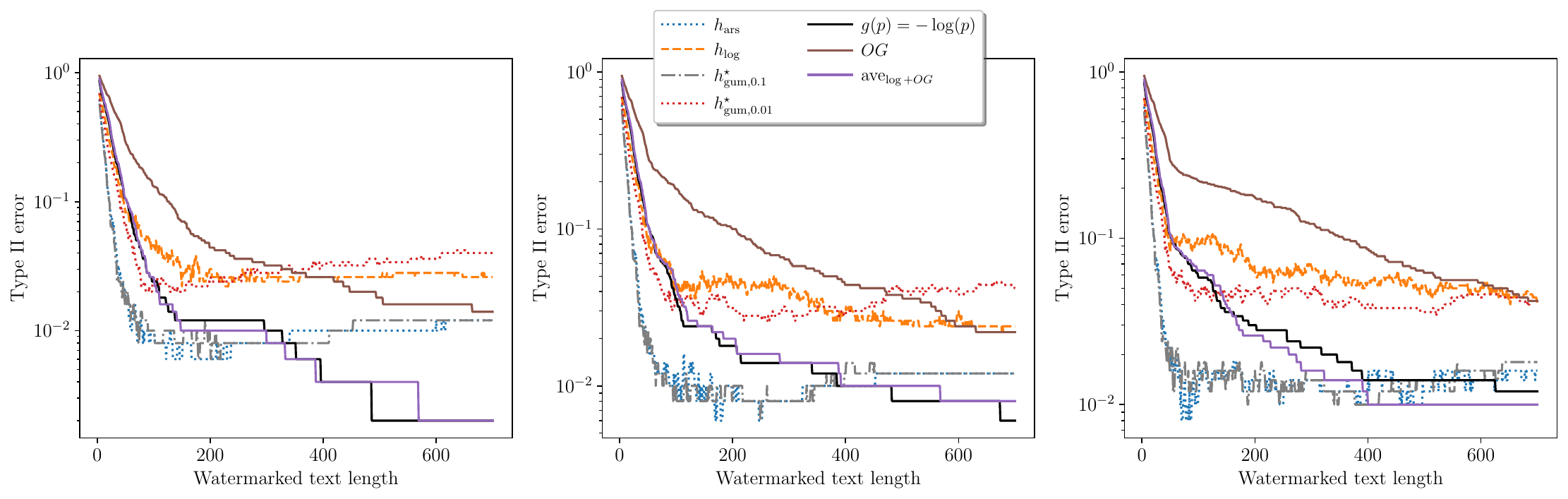}
    \caption{Average Type II error (log scale) versus text length on the OPT-1.3B model with temperature 0.5 under human editing after the first 50 tokens. The editing proportions are 0.1 (left), 0.3 (middle), and 0.5 (right), respectively.}
    \label{fig2}
\end{figure}

Together with the simulation results in Figure \ref{fig0}, we have the following findings. 
\begin{enumerate}
\item  Only the e-process methods demonstrate robust Type I error control in both fixed-sample and sequential testing setups. In contrast, the sum-based methods can only provide (asymptotic) Type I error control in the non-sequential setup. 
In the setting of sequential detection, the Type I error quickly explodes, making these non-e-process methods unsuitable for the online testing task.
\item   Some sum-based methods, especially $h_{\text{ars}}$,   exhibit better empirical power (one minus Type II error rate) compared to the e-process methods  in both the simulated setting and the LLM-generated-text setting. The superior power is expected, as these methods offer weaker statistical validity guarantee.
\item Among the three e-process methods,
the OG e-process 
and the average e-process 
perform similarly  in the simulated setting, visibly better than the weight-adaptive e-process. The weight-adaptive e-process and the average e-process perform similarly in the LLM-generated-text setting, visibly better than the OG  e-process. We would therefore recommend the average e-process in practice. 
\item 
The average e-process method
has roughly comparable power to the sum-based methods. 
Perhaps surprisingly,
while providing sequential validity, it can even have higher power than the best sum-based methods in some settings of token length and temperature; see the top panel of Figure \ref{fig1}. 

\item  While Type II error rates for all methods tend to decrease as the text sequences grow longer, sum-based methods exhibit a slight inflation in Type II error rates under low-temperature settings. 
This inflation in sum-based methods likely stems from the emergence of degenerate NTPs at lower temperatures, particularly in extended sequences—a phenomenon documented in recent studies \citep{pee24temperature,LI25temperature}. Since tokens generated from degenerate distributions are deterministic and independent of the watermark's pseudo-randomness, they act as noise that interferes with sum-based detection. In contrast, e-process methods maintain a consistent and monotonic reduction in Type II errors, making them well suited to such situations.

\item As the proportion of human editing increases, all methods exhibit higher Type II error rates for a given text length under the same model temperature setting. Despite this, our methods still achieve power comparable to existing approaches.
\end{enumerate}

\section{Conclusion}
\label{sec:conclude}

In this paper, we presented a unified framework for LLM watermark detection based on e-processes, providing anytime-valid guarantees both theoretically and empirically. 
We show in Theorem \ref{th:revese} that our proposed method is the only admissible and unbiased method for this testing problem under some mild assumptions, satisfied by, e.g., the Gumbel-max watermark. 
Various empirically adaptive e-process construction methods are developed to improve the testing power. In addition,  our results analyze the power properties of the proposed procedures, illustrating their statistical efficiency and asymptotic behavior. Experimental evaluations further demonstrate the advantages of the proposed framework over existing watermark detection methods. 

The present study remains preliminary in several aspects. The theoretical characterization of optimal e-process constructions is not yet complete, especially with respect to achieving the theoretical upper bounds on the e-power. Furthermore, the selection and design of statistically efficient pivotal statistics tailored to different watermark mechanisms warrant deeper investigation. Future work may focus on refining the theoretical limits of sequential testing power within this framework, exploring adaptive and data-driven construction strategies, and developing principled criteria for selecting optimal test statistics under varying watermark schemes. These directions are expected to further strengthen the statistical foundations and practical effectiveness of LLM watermark detection.






\bigskip



\newpage

\begin{center}

{\bf\large Appendix: Proofs}

\end{center}

\begin{proof}[Proof of Lemma \ref{lem:gm}]
   By  Proposition C.1 of \cite{marshall2011inequalities}, the sum of univariate convex functions is Schur-convex. 
 Therefore, for statement (i), it suffices to prove that $p\mapsto py^{1/p}$ is convex on $[0,1]$, which follows by straightforward calculus. 
For statement (ii), it suffices to note that $\ntp\preccurlyeq  \ntp^*$
for all $\ntp\in \Delta^{\delta}_K$, which can be checked by definition.
Statement (iii) follows directly from    $F^{\ntp}(y) \le F^{\ntp^*}(y) $ for any $y\in [0,1]$ and $\ntp\in \Delta_K^{\delta}$, as shown in (ii). 
\end{proof}

\begin{proof}[Proof of Proposition \ref{prop:KLent}]
Straightforward calculation yields 
\begin{align*}
\mathrm{KL}(\q^{\ntp}\Vert \uu^\ntp) 
&= \E^{\q^\ntp}\left[ 
\log\frac{\d \q^{\ntp}}{\d  \uu^\ntp} 
\right] 
=\sum_{w\in \mathcal W} P_w  
\E^{\q^\ntp}\left[
\log \frac{\d \q^{\ntp} }{\d \uu^\ntp  }\mid W=w\right].
\end{align*}
For any $w$, the distribution of $\zeta$ under $\q^\ntp$ given $W=w$
is uniform on $ \{ z\in Z: S(\ntp,z )=w\}$, whereas the distribution of $\zeta$ under
$\uu^{\ntp}$ given $W=w$ is uniform on $Z$. 
Therefore, 
  $$\q^\ntp \left (\frac{\d \q^\ntp}{\d\uu^\ntp}=\frac{1}{P_w} \mid W=w\right)=1 .$$
Therefore, 
\begin{align*}
 \sum_{w\in \mathcal W} 
 P_w
\E^{\q^\ntp}\left[
\log \frac{\d \q^{\ntp} }{\d \uu^\ntp  } \mid {W=w} \right]
 =  
 -\sum_{w\in \mathcal W} P_w \log P_w ,
\end{align*}
and hence, $\mathrm{KL}(\q^{\ntp}\Vert \uu^\ntp) = \ent(\ntp)$.  
\end{proof}

\begin{proof}[Proof of Theorem \ref{th:validity}]
We first verify that $E_1,E_2,\dots$ are sequential e-values for the null hypothesis. 
We have 
$$
\E^\p[E_t|\mathcal F_{t-1}] = \E^\p[f_t(1-Y_t)|\mathcal F_{t-1}] = \int_0^1 f_t (p) \d p = 1,
$$
where we have used that $Y_t $ is independent of $\mathcal F_{t-1}$
and that $f_t $ is determined by $Y_1,\dots,Y_{t-1}$ and it integrates to $1$.
Moreover, it is clear that $E_t$ is nonnegative. 
Therefore,
$$
\E^\p[M_t|\mathcal F_{t-1}] = M_{t-1}, ~~~t\in \TP.
$$
Thus, $M$ is a nonnegative  martingale starting at $1$ and hence an e-process.  The probability guarantee follows from Ville's inequality or the optional stopping theorem. 

For the last statement, it suffices to note that, for each $\q$ in the alternative hypothesis, using the superuniformity of $Y_t$ and the fact that $f_t$ is decreasing, we have  $\E^\q[f_t(1-Y_t)|\mathcal F_{t-1}] \ge 1$ 
for all $t\in \TP$.
Therefore, $\E^\q[M_t]\ge 1$. 
 \end{proof}

\begin{proof}[Proof of Theorem \ref{th:revese}]
The ``if'' statement follows from Theorem \ref{th:validity}, and we show the ``only if'' statement.  
  By \citet[Theorem 18]{ramdas2020admissible}, an admissible e-process $M$ is a $\p$-nonnegative martingale  with $M_0=1$.    Let $E_t=M_t/M_{t-1}$ with the convention $0/0=1$.
Note that $\mathcal F_t=\sigma(Y_1,\dots,Y_t)$, and hence $E_t$ can be written as a function $e_t$ of $(Y_1,\dots,Y_t)$. Since $M'$ is a nonnegative martingale and $Y_1,\dots,Y_t$ are independent, we have
$$
\E^\p[e_t(Y_1,\dots,Y_t)|\mathcal F_{t-1}] =\int_0^1 e_t(Y_1,\dots,Y_{t-1},y)  \d y = 1.
$$
Let  $f_t:y\mapsto e_t(Y_1,\dots,Y_{t-1},1-y)$ for $t\in \TP$, which  is Borel measurable in $y$.
We next show, for each $t$, 
$$\p\left(\{\mbox{$f_t$
is decreasing almost everywhere on $(0,1)$}\}\right)=1.$$
The intuition is that if this does not hold true, then a larger $Y_t$ may not lead to a larger $E_t$, which yields $\E^\q[M_t]<1$ for some $\q$ in  the alternative hypothesis. 

Suppose for contradiction that there exists $t\in \TP$ such that the set $$A: =\{\mbox{$f_t$
is not decreasing almost everywhere on $(0,1)$}\} \in \mathcal F_{t-1} $$
has positive probability.  
For every value of $\mathbf y\in[0,1]^{t-1}$   taken by $\mathbf Y:=(Y_1,\dots,Y_{t-1})$ on the event $A$, there exists $x^\mathbf y\in(0,1)$ such that the set
$$B^\mathbf y:=\{y\in (x_0,1): e_t(\mathbf y, y)
>e_t(\mathbf y, x^\mathbf y)\}$$ has positive  Lebesgue measure.

We will now construct a probability measure $\q$ in the alternative hypothesis through the distribution of $Y_t$ under $\q$. Comparing $\p$ and $\q$, we will only modify the distribution of $Y_t$ on the event $A$, and everything stays unchanged outside $A$. 
Let 
$U$ be a uniform random variable on $[0,1]$ independent of everything else. 
Define the function $h^{\mathbf y}(u)=x_0 \id_{\{u\in B^{\mathbf y}\}}
+u\id_{\{u\not \in B^{\mathbf y}\}} 
$ for
$y\in (0,1)$.
It is clear that  $h^{\mathbf y}(u)\le u$ for $u\in (0,1)$ and hence $1-h^{\mathbf y}(U)$ is superuniform. 

Under $\q$, we let  $Y_t$ be distributed as  $1-h^{\mathbf Y}(U)$ on event $A$, which is superuniform; $(Y_s)_{s\in\TP}$ remains an independent sequence with  $Y_s$  uniformly distributed  for $s\ne t$. 
The measurability of $Y_t$ (to be well-defined as a random variable) is guaranteed by the assumption of independence. 
Because $B^{\mathbf y}$ has positive Lebesgue measure, $e_t(\mathbf y, y)>e_t(\mathbf y,x^{\mathbf y})$ for $y\in B^{\mathbf y}$ and $\int_0^1 e_t(\mathbf y, u) \d u=1$, we have 
$$
\int_0^1 e_t(\mathbf y, h^{\mathbf y}(u))\d u <\int_0^1 e_t(\mathbf y, u)\d u  =1.
$$
Therefore, 
\begin{align*}
\E^\q[f_t(1-Y_t)|\mathcal F_{t-1}]&=
\E^\q[f_t(h^{\mathbf Y}(U))|\mathcal F_{t-1}]\id_A+\E^\q[f_t(  U)|\mathcal F_{t-1}]\id_{A^c}\\&=
\left(\int_0^1 e_t(\mathbf Y, h^{\mathbf Y}(u))\d u\right) \id_{A} + \id_{A^c}.
\end{align*}
Since $Z_t:=\int_0^1 e_t(\mathbf Y, h^{\mathbf Y}(u))\d u<1$ on $A$ and $A$ has positive probability under both $\p$ and $\q$,   
we obtain
\begin{align*}
\E^\q[M_t]&=
\E^\q\left[M_{t-1}\E^\q[f_t(h^{\mathbf Y}(U))|\mathcal F_{t-1}]  \right]
 =
\E^\q\left[ M_{t-1} Z_t \id_A+ M_{t-1}\id_{A^c}\right] <\E^{\q}[M_{t-1}]=1. 
\end{align*}
 Therefore, $\E[M_t]<1$, a contradiction. Therefore, $A$ must have zero probability, and \eqref{eq:dominance} holds with $f_t\in \Phi_t$ for each $t$. 
\end{proof}

\begin{proof}[Proof of Theorem \ref{th:fixed}]
Let $E_t=g(1-Y_t)$ for $t\in \TP$. The distribution of $E_t$ depends on $\ntp_t$ through \eqref{eq:FPY}, and under $\q$ the sequence $(\log (1-\lambda+\lambda E_t)_{t\in \TP}$ is $m$-dependent.  
Let $\lambda$ be taken as $\lambda_0>0$ in 
  Lemma \ref{lem:E1} below.
  Using that lemma,
there exists a constant $\epsilon>0$ such that  $$\E^{\q}[ \log (1-\lambda+\lambda E_t) \mid  \ntp_t ] \ge  \epsilon. 
$$ 
Note that, since $\d F^{\ntp}/\d \mathbb{U}\le K$ by \eqref{eq:FPY-den} and $Y_t$ is uniformly distributed under $\p$, we have 
\begin{align*}
    \E^{\q}[(\log (1-\lambda+\lambda E_t))^2\mid \ntp_t]
    &\le K \E^{\p}[(\log (1-\lambda+\lambda E_t))^2].
\end{align*}
Since $1-\lambda+\lambda E_t$, $t\in \TP$ are identically distributed under $\p$ with finite mean,  $\E^{\p}[(\log (1-\lambda+\lambda E_t))^2]$ is bounded above across $t$.
Hence, $  \E^{\q}[(\log (1-\lambda+\lambda E_t))^2]$ is also bounded above.
This allows us to use the strong law of large numbers for $m$-dependent sequences (e.g., Theorem 1 of \citealp{hansen1991strong}),
which gives  
$$\frac 1T \left(\log M_T -\E^{\q^\ntp}[\log M_T]\right) \to 0 \mbox{ almost surely.}$$
Consequently,  
$ 
    \liminf_{T\to\infty}  ( \log M_T )/T
   \ge \epsilon >0
 $
almost surely.
\end{proof}

\begin{proof}[Proof of Lemma \ref{lem:E1}]
Let  $U$ be uniformly distributed on $[0,1]$ under $\p$. 
By Lemma \ref{lem:gm}, 
for any $\bP\in \Delta_K^{\delta}$, we have 
     \begin{align}
\min_{\bP\in \Delta_K^{\delta}} \E^{\q^{\bP}} [ \log (1-\lambda+\lambda g(1-Y))  ]
= \E^{\q^{\bP^*}} [ \log (1-\lambda+\lambda g(1-Y))   ], \label{eq:lem-E1}
\end{align}
where $\ntp^*=(1-\delta,\delta,0,\dots,0)$. Note that for any $\ntp\in \Delta^\delta_K$, 
     \begin{align*}
\E^{\q^{\bP}} [ g(1-Y)  ]& = \E^{\p}\left[\sum_{w\in \mathcal W} U^{1/P_{w}-1}  g(1-U)   \right]
\\&  > \E^{\p}\left[\sum_{w\in \mathcal W} U^{1/P_{w}-1}  \right]
\E^{\p}\left[ g(1-U) \right]
 =  1,
\end{align*}
where the inequality follows from Chebyshev's association inequality, using that $g$ is nonconstant. 
This yields, in particular, $\E^{\q^{\bP^*}} [g(1-Y)]>1$. By using Theorem 3.14 of \cite{RW24},
this further implies 
\begin{equation}
\label{eq:lambda0}\E^{\q^{\bP^*}} [\log (1-\lambda_0+\lambda_0 g(1-Y))]>0
\end{equation} for some  $\lambda_0>0$.
 Since the function $\lambda \mapsto \log (1-\lambda+\lambda g(1-Y))$ is concave and takes the value $0$ at $\lambda=0$, 
 we know that  it is positive  on $ (0,\lambda_0].$
Therefore, using \eqref{eq:lem-E1}, 
we get the desired statement.  
\end{proof}

\begin{proof}[Proof of Theorem \ref{th:OG}]
First, consider the weight-adaptive e-process $M$ in \eqref{eq:EAE}.
For $t\in \TP$, let $E_t=g(1-Y_t)$.
 Assumption \ref{assm:power}  guarantees $\E^{\q}[E_t]>1$.
Let
$$
\lambda^* =\argmax_{\lambda \in [0,\gamma]}\E^{\q}[\log (1-\lambda +\lambda E_t)].
$$
Note that  the optimizer  $\lambda^*$ is unique because of the strict concavity of the log function,   $\lambda^*$ does not depend on $t$ because $E_1,E_2,\dots$ are iid, and $\lambda^*>0$, which follows from Theorem 3.14 of \cite{RW24} and  the condition $\E^{\q}[E_t]>1$. 
Note that \begin{equation}
    \label{eq:Lstar}
    \E^{\q}[\log (1-\lambda^* +\lambda^* E_1)] >0
    \end{equation} because $\lambda^*>0$ and it is finite, as  guaranteed by Assumption \ref{assm:power}.
Let 
$$
    M^*_t=
\prod_{s=1}^t (1-\lambda^*+\lambda^*E_s),~~~t\in \T. 
$$
Proposition EC.2 of \cite{wang2025backtesting}
gives
$$
\frac{1}{T}(\log M_T- \log M^*_T)=0 \mbox{~~~in $L^1$ under $\q$.}
$$
By the strong law of large numbers, 
$$
\lim_{T\to\infty} \frac 1T \log M^*_T  = \E^{\q}[\log (1-\lambda^* +\lambda^* E_1)]>0 \mbox{~almost surely.}
$$
 Thus, we get the desired  exponential growth of the weight-adaptive e-process.  
 
To show the statement on the OG e-process $M$ in \eqref{eq:OGE} with range $[a,b]$, $0<a<1<b<\infty$, we use Theorem 3.1 of \cite{hore2026online},
which gives
$$
\frac{1}{T}(\log M_T- \log \widetilde M_T)=0 \mbox{~~~almost surely,}
$$
where $\widetilde M_t=
\prod_{s=1}^t  f^*(1-Y_t),~t\in \T,
$
and $f^*$ maximizes $\E^{\q}[\log f(1-Y_t)]$ over all decreasing densities in $\mathcal D$ with range $[a,b]$. Take any nonconstant calibrator $g$ that takes values in $[a,b]$, and let $\lambda^*>0$ be as in the first part of the proof, which satisfies \eqref{eq:Lstar}. Since $1-\lambda^* +\lambda^* g$
falls in $[a,b]$, we have
$$
\E^\q[\log f^*(1-Y_t)]
\ge \E^{\q}[\log (1-\lambda^* +\lambda^* g(1-Y_t))]>0.
$$
Therefore, the law of large numbers gives 
 $$ \lim_{T\to\infty} \frac 1T \log \widetilde M_T =\E^\q[
\log f^*(1-Y_t)]  >0~~~\mbox{almost surely.}$$
This yields the desired statement on the OG e-process (even stronger, almost sure convergence).

The last statement on the average e-process follows directly from the corresponding statement on the weight-adaptive e-process.
\end{proof}

\newpage

\bibliographystyle{apalike}
 \bibliography{ref} 
\end{document}